\documentclass[a4paper,UKenglish,cleveref, autoref, thm-restate]{lipics-v2021}

\usepackage[dvipsnames,usenames]{xcolor}

\RequirePackage{hyperref}
\hypersetup{colorlinks=true, urlcolor=Blue, citecolor=Green, linkcolor=BrickRed, breaklinks, unicode}

\usepackage[noadjust]{cite}

\usepackage{cite}

\usepackage{graphicx}
\usepackage{amssymb}
\usepackage{todonotes}

\usepackage{epstopdf}
\def\eps{\varepsilon}
\def\poly{poly}

\newcommand{\celldegree}{\text{cell-degree}}
\def\Vor{{ \rm Vor}}
\def\VD{{ \rm VD}}
\newcommand{\para}[1]{\vspace{-0.1in} \subparagraph*{#1}}

\title{What Else Can Voronoi Diagrams Do For Diameter In Planar Graphs?
}
\author{Amir Abboud}{The Weizmann Institute, Israel}{amir.abboud@weizmann.ac.il }{https://orcid.org/0000-0002-0502-4517}{Supported by an Alon scholarship and a research grant from the Center for New Scientists at the Weizmann Institute of Science.}

\author{Shay Mozes}{Reichman University, Israel}{smozes@idc.ac.il}{https://orcid.org/0000-0001-9262-1821}{Israel Science Foundation grant 810/21.}

\author{Oren Weimann}{University of Haifa, Israel}{oren@cs.haifa.ac.il}{https://orcid.org/0000-0002-4510-7552}{Israel Science Foundation grant 810/21.}

\authorrunning{A. Abboud, S. Mozes, and O. Weimann}

\Copyright{Amir Abboud, Shay Mozes, and Oren Weimann}



\date{}

\category{} 

\relatedversion{}

\nolinenumbers 

\hideLIPIcs 

\EventEditors{Martin Farach-Colton, Inge Li G\o rtz, and Simon J. Puglisi }
\EventNoEds{3}
\EventLongTitle{31st European Symposium on Algorithms (ESA 2023)}
\EventShortTitle{ESA 2023}
\EventAcronym{ESA}
\EventYear{2023}
\EventDate{September 4-8, 2023}
\EventLocation{Amsterdam, Netherlands}
\EventLogo{}
\SeriesVolume{}
\ArticleNo{}

\begin{document}
\maketitle
\thispagestyle{empty}

\begin{abstract}

The Voronoi diagrams technique, introduced by Cabello [SODA'17] to compute the diameter of planar graphs in subquadratic time, has revolutionized the field of distance computations in planar graphs.
We present novel applications of this technique in static, fault-tolerant, and partially-dynamic undirected unweighted planar graphs, as well as some new limitations. 

\begin{itemize}
\item In the static case, we give $n^{3+o(1)}/D^2$ and $\tilde O(n\cdot D^2)$ time algorithms for computing the diameter of a planar graph $G$ with diameter $D$. These are faster than the state of the art $\tilde O(n^{5/3})$ [SODA'18] when $D<n^{1/3}$ or $D>n^{2/3}$. 

\item
In the fault-tolerant setting, we give an $n^{7/3+o(1)}$ time algorithm for computing the diameter of $G\setminus \{e\}$ for every edge $e$ in $G$ (the replacement diameter problem). This should be compared with the naive $\tilde O(n^{8/3})$ time algorithm that runs the static algorithm for every edge.
\item
In the incremental setting, where we wish to maintain the diameter while adding edges, we present an algorithm with total running time $n^{7/3+o(1)}$.
This should be compared with the naive $\tilde O(n^{8/3})$ time algorithm that runs the static algorithm after every update.
\item
We give a lower bound (conditioned on the SETH) ruling out an amortized $O(n^{1-\eps})$ update time for maintaining the diameter in {\em weighted} planar graph. The lower bound holds even for incremental or decremental updates. 
 \end{itemize} 
 
Our upper bounds are obtained by   novel uses and manipulations of Voronoi diagrams. These include maintaining the Voronoi diagram when edges of the graph are deleted, allowing the sites of the Voronoi diagram to lie on a BFS tree level (rather than on boundaries of $r$-division), and a new reduction from incremental  diameter to incremental \emph{distance oracles} that could be of interest beyond planar graphs. 
Our lower bound is the first lower bound for a dynamic planar graph problem that is conditioned on the SETH.

\end{abstract}

\newpage
 
\clearpage
\setcounter{page}{1}

\section{Introduction}

The {\sc diameter} problem asks to compute the largest distance in the graph. 
It is one of the most basic and extensively studied problems in the graph algorithms literature, and moreover, it is prominent in Fine-grained Complexity where it has driven the development of innovative hardness reductions \cite{RV13,ChechikLRSTW14,AbboudVassilevskaFOCS14,AWW16,AbboudCKP21,AnconaHRWW19,BackursRSWW21,Bonnet22,DalirrooyfardLW21}.
Assuming the strong exponential time hypothesis (SETH), there is also no truly subquadratic algorithm for {\sc diameter}~\cite{RV13,AWW16} in  undirected, unweighted graphs with treewidth $\Omega(\log n)$. For graphs of bounded treewidth, the diameter can be computed in near-linear time~\cite{AWW16} (see also~\cite{Husfeldt16,Eppstein} for algorithms with time bounds that depend on $D$). Near-linear time algorithms were developed for many other restricted graph families, see e.g.~\cite{Chordal,PlaneTriangulations,CactusGraphs,OuterplanarGraphs,IntervalGraphs,HyperbolicGeodesic,Euclidian,SMAWK,Geodesic,DistanceHereditaryGraphs}.

One of the outstanding questions that has remained open despite a decade of major developments in algorithms and conditional lower bounds for graph problems is whether {\sc diameter} can be solved in near-linear time in \emph{planar graphs}.
Until 2017, only logarithmic improvements over the natural $O(n^2)$ bound (of computing all-pairs shortest-path, APSP) had been known~\cite{Chan12,WN08}.
The consensus was that truly subquadratic time is impossible and the focus of the community was on proving a hardness result, e.g. under SETH.
But then, in a celebrated paper, Cabello~\cite{Cabello} gave a subquadratic $\tilde O(n^{11/6})$ time algorithm, that was later improved to the current-best $\tilde{O}(n^{5/3})$ bound~\cite{DiameterSODA18}.

The breakthrough in Cabello's work~\cite{Cabello} is his novel use of \emph{Voronoi Diagrams} (VDs) in planar graph algorithms. 
This new machinery has revolutionized the field of distance computation problems in planar graphs and has lead to several breakthroughs~\cite{Cohen-AddadDW17,ourSODA2018,faultyOracle,Charalampopoulos19,LongPettie} including a surprising and almost-optimal \emph{distance oracle} - a problem that had hitherto seen many gradual improvements using different techniques both in the exact~\cite{ArikatiCCDSZ96,Djidjev96,ChenX00,FR06,K05,Wulff-Nilsen10,Nussbaum11,Cabello12,MozesS2012,Cohen-AddadDW17,ourSODA2018,Charalampopoulos19,LongPettie} and the approximate~\cite{Thorup04,KawarabayashiKS11,KawarabayashiST13,Klein02,GuX19,Wulff-Nilsen16,ChanS19} settings.  
Consequently, the main meta question occupying the minds of researchers in planar graph algorithms is: \emph{what else can Voronoi diagrams do for us?}

\subsection{Dynamic Planar Diameter}
It is natural to expect VDs to produce breakthroughs in the domain of \emph{dynamic} planar graphs.
Dynamic data structures that support updates and queries to a graph have remarkable applications in theory (as a subroutine in static algorithms) and practice (for changing inputs).
Many ingenious algorithms for basic problems in dynamic planar graphs have been developed in the last few decades, including connectivity, distances, and cuts ~\cite{K05,FR06,DasGW22,ChangKT22,insw-iamcm-11,LS11,AbrahamCG12,BorradailePW12,BorradaileSW10,KawarabayashiST13,Thorup04,ubramanian93,faultyOracle,LackiS17,ItalianoKLS17,Karczmarz18}, but large (polynomial) gaps remain compared to the lower bounds \cite{AD16}.
Only few of these works~\cite{faultyOracle,Charalampopoulos22} use VDs and only in a limited way (they recompute the VD from scratch after every update).
It is clear that major advancements await if one is able to maintain the VD machinery \emph{dynamically} in a meaningful way.
In this paper, we investigate this possibility by focusing on the {\sc diameter} problem.

The state-of-the-art algorithm recomputes the diameter from scratch after every update in time $\tilde O(n^{5/3})$. 
This is not surprising since the only useful technique against {\sc diameter} (in static graphs) is based on VDs, and we do not know how to make VDs dynamic.

The first question that comes to mind is: Suppose, optimistically, we could make VDs as dynamic as possible; \emph{what time bound would we hope to get?} 
Clearly, we cannot get $O(n^{2/3-\eps})$ time per update until we break the $\tilde O(n^{5/3})$ bound for static graphs.
Moreover, a conditional $n^{2/3-o(1)}$ lower bound (under the APSP or Online Matrix Vector Conjectures) follows from the reductions of Abboud and Dahlgaard~\cite{AD16}.
So perhaps dynamic VDs would lead to a matching $O(n^{2/3})$ upper bound?
Our first result rules out this possibility with an $n^{1-o(1)}$ lower bound under SETH.

\begin{theorem}[Lower Bound on Dynamic Diameter]\label{thm:lower}
If the diameter of a dynamic undirected planar graph on $n$ nodes can be maintained with $O(n^{1-\eps})$ amortized time per weight-change, then SETH is false.
This holds even if the dynamic algorithm is allowed to preprocess the initial graph in $\poly(n)$ time, and even in the partially-dynamic setting where weights only increase or only decrease.
\end{theorem}

Notably, this is the first lower bound for a dynamic planar graph problem that is based on the SETH (as opposed to other conjectures) and only the second example of such a result if we consider \emph{static} planar graph problems as well \cite{ACK20,GawrychowskiMW21}.

\para{Towards Dynamic Voronoi Diagrams.}
A large gap of $n^{2/3}$ remains despite our lower bound and it is likely that it can be closed if we can indeed make VDs dynamic.\footnote{It is tempting to think that \cref{thm:vor} implies a dynamic diameter algorithm with update time $\tilde{O}(n^{1.6})$; Use an $r$-division and maintain for each piece the DDG and bisectors. Upon an update of an edge in a piece $P$, recompute the DDG of $P$ (using MSSP) and the bisectors of $P$ (using \cref{thm:vor}). For each vertex in the graph, recompute all additive weights using FR-Dijkstra, and compute the furthest vertex in each piece using \cref{thm:vor}. The caveat is that this approach does not handle  properly the case where both endpoints of the diameter path belong to the same piece (not necessarily $P$). The reason is that the VD mechanism only handles paths that visit at least one boundary node.}
In this paper, we take a small (but arguably the first) step towards this goal: we give an efficient algorithm for updating the VD after the deletion of one edge in the graph, much faster than recomputing it from scratch.
(We refer to Section~\ref{sec:replacement} for an overview and all the details.)
This small step already has interesting applications. While it applies for general (weighted) planar graphs, the applications we have found only gain an advantage in unweighted planar graphs.  

A concrete application is a faster algorithm for the {\sc replacement diameter}: given a graph $G$ return the diameter of $G\setminus \{e\}$, the graph obtained by removing the edge $e$, for all edges $e$.
The trivial algorithm for this problem makes $O(n)$ calls to a static diameter algorithm, one for each edge, and achieves $\tilde{O}(n^{8/3})$ running time.
We improve this upper bound by an $n^{1/3}$ factor to $n^{7/3+o(1)}$ by utilizing our efficient updates to VDs, along with other tricks that are also based on VDs (but not in a dynamic way).

\begin{theorem}[Replacement Diameter]\label{thm:replacement}
   Given an unweighted undirected planar graph $G=(V,E)$, there is an  $n^{7/3+o(1)}$ time algorithm that for every edge $e \in E$ outputs the diameter of $G^{e}=(V, E \setminus \{e\})$. 
\end{theorem}

An additional new result is a faster algorithm for {\sc diameter} in the \emph{incremental} setting where we start from an empty graph and need to maintain the diameter while $O(n)$ edges are being added (without violating the planarity).
The trivial algorithm recomputes the diameter after every update in a total of $\tilde{O}(n^{8/3})$ time, and we improve it to $n^{7/3+o(1)}$.

\begin{theorem}[Incremental Diameter]\label{thm:incremental}
There is an algorithm that maintains the diameter of an unweighted undirected planar graph undergoing edge insertions in a total of $n^{7/3+o(1)}$ time.
\end{theorem}

This result is based on an elegant reduction from incremental {\sc diameter} to incremental \emph{distance oracles} that could be of interest beyond planar graphs. 
Its analysis relies on recent works on the \emph{bipartite independent set} queries introduced by Beame et al. \cite{beame2020edge}.

\subsection{Static Planar Diameter}

Back to {\sc diameter} in static graphs, what else can we hope to get from VDs?
Of course, the biggest open question is whether the $n^{5/3}$ bound can be improved to $n^{1+o(1)}$, or whether one can prove a super-linear lower bound.
Toward this question, we would like to understand the hard/easy cases, and a natural parameter to consider is $D$ -- the diameter itself.

One of the main algorithmic contributions of this paper, that is crucial to the aforementioned upper bounds, is an algorithm beating $n^{5/3}$ when $D$ is large (in the range $[n^{2/3+\eps},n]$).
Notably, it implies that anyone seeking a tight conditional lower bound cannot use constructions with very large diameter.

\begin{theorem}[Static Large Diameter]\label{thm:staticlargediameter}
    The diameter can be computed in  $n^{3+o(1)}/D^2$ time on an unweighted undirected planar graph with diameter $D$.  
\end{theorem}

Our new algorithm applies VDs in a novel way, where the VD sites lie on a BFS tree level, as opposed to lying on the boundary of pieces in an $r$-divisions. 

While our result is the first to address the large $D$ case, the other extreme of small $D$ has already been studied.
Eppstein~\cite{Eppstein} gave the first near-linear time algorithm for constant $D$, with an exponential dependence on $D$. 
This dependency was later improved as a byproduct of new $(1+\eps)$-approximation algorithms for {\sc diameter}~\cite{ChanS19,WY16,Eppstein,BermanKasiviswanathan}.
The state of the art is $\tilde O(n \cdot D^5 )$ using the $(1+\varepsilon)$-approximation $\tilde O(n \cdot (1/\varepsilon)^5)$-time algorithm of Chan and Skrepetos \cite{ChanS19} with $\varepsilon = 1/D$.
The final result of this paper is an improved bound of $\tilde{O}(nD^2)$ which increases the range in which the $n^{5/3}$ bound can be beaten from $D< n^{2/15-\eps}$ to $D<n^{1/3-\eps}$.

\begin{theorem}[Static Small Diameter]\label{thm:staticsmalldiameter}
    The diameter can be computed in   $\tilde O(n \cdot D^2)$ time on an unweighted undirected planar graph with diameter $D$.  
\end{theorem}

Our algorithm exploits VDs in a more natural way than that of Chan and Skrepetos \cite{ChanS19}, if our goal is solve the small $D$ case exactly (recall that their focus is on approximations).
It remains an interesting open question whether the $\tilde O(n \cdot (1/\varepsilon)^5)$ time approximation algorithm can be improved.
This is related to another challenge of computing \emph{approximate VDs} faster than exact, which we do not address in this paper.


\section{Preliminaries}

A recursive decomposition tree $\mathcal{T}$ of a planar graph $G$ is the tree obtained (in linear time) by recursively separating $G$ with a separator of size $\sqrt{|G|}$. $\mathcal{T}$ is a binary tree whose nodes correspond to subgraphs of $G$ (called {\em pieces}), with the root being all of $G$ and the leaves being pieces of constant size. We identify each piece $P$ with the node representing it in $\mathcal{T}$ (we can thus abuse notation and write $P\in \mathcal{T}$), and with its boundary $\partial P$ (i.e.~vertices that belong to some separator along the recursive decomposition used to obtain $P$). An important property for us (see e.g. \cite[Lemma 3.1]{ourSODA2018}) is that the sum of $|P| \cdot |\partial P'|$ over all pairs of siblings $P,P'$ in $\mathcal T$ is $\tilde O(n^{1.5})$.

An $r$-{\em division}~\cite{F87} of a planar graph $G$ is a decomposition
of $G$ into $\Theta(n/r)$ pieces,
each of them with $O(r)$ vertices and $O(\sqrt{r})$ boundary vertices (vertices shared with other pieces). It is possible to compute an $r$-division in $O(n)$ time~\cite{KMS13} with the useful property that the boundary vertices of each piece lie on a constant number of faces of the piece (called {\em holes}).

The \emph{dense distance graph} (DDG) of a piece $P$ is the complete graph over the boundary vertices of $P$. The length of edge 
$uv$ in the DDG of $P$ equals to the $u$-to-$v$ distance inside $P$. Note that the DDG of $P$ is non-planar. The DDG of an $r$-division is the union of DDGs of all pieces of the $r$-division. 
Thus, the total
number of vertices in the DDG is  $O(n/\sqrt r)$, and the total
number of edges is $O(n)$.
The DDG of an $r$-division can be computed in $\tilde O(n)$ time  
using the MSSP algorithm~\cite{K05}.
Fakcharoenphol and Rao~\cite{FR06} described an $\tilde O(n/\sqrt{r})$ time implementation of
Dijkstra's algorithm (nicknamed FR-Dijkstra) on the DDG. 

The difficult case for computing the diameter is when the furthest pair of vertices lie in different pieces. Consider some source vertex $s$ outside of some piece $P$. For every boundary vertex $u$ of $P$, let $d(u)$ denote the $s$-to-$u$ distance in $G$. The {\em additively weighted Voronoi diagram} of $P$ with respect to $d(\cdot)$ is a partition of the vertices of $P$ into pairwise disjoint sets (Voronoi cells), each associated with a unique boundary vertex (site) $u$. The vertices in the cell
$\Vor(u)$ are all the vertices $v$ of $P$ such that $u$ is the last boundary vertex of $P$ on the shortest $s$-to-$v$ path. In other words, every site $u$ of $P$ has {\em additive weight} $d(u)$, the {\em additive distance} from a site $u$ to a vertex $v$ of $P$ is defined as $d(u)$ plus the length of the shortest $u$-to-$v$ path inside $P$, and the cell $\Vor(u)$ contains all vertices $v$ of $P$  that are closer (w.r.t. additive distances) to $u$ than to any
other site in $S$. 
The {\em boundary} $\partial \Vor(u)$ of a cell $\Vor(u)$ consists of all edges of $P$ that have exactly one endpoint in $\Vor(u)$. For example, in a Voronoi diagram of just two sites $u$ and $v$, the boundary of the cell $\Vor(u)$ is a $uv$-cut and is therefore a cycle in the dual graph. This cycle is called the $uv$-{\em bisector}. 
The complexity $|\partial \Vor(u)|$ of a Voronoi cell $\Vor(u)$ is the number of faces of $P$ that contain vertices of $\Vor(u)$ and of at least two more Voronoi cells. 
For every source $s$, computing the furthest vertex from $s$ in $P$ thus boils down to computing, for each site $u$, the furthest vertex (w.r.t. additive distance) from $u$ in $\Vor(u)$, and  then returning the maximum value among all sites $u$. 

\begin{theorem}[\cite{DiameterSODA18}]\label{thm:vor}
Let $P$ be an edge-weighted planar graph with $r$ vertices. Let $S$ be a set of $b$ sites that lie on the boundaries of $\tilde O(1)$  faces\footnote{Theorem 1.1 in ~\cite{DiameterSODA18} is phrased for a constant number of faced (called holes). However, as pointed in footnote~8 in~\cite{DiameterSODA18}, the dependency of the running time on the number of holes is polynomial, so the theorem applies also to the case of a polylogarithmic number of holes.}
of $P$. The $uv$-bisectors of all pairs $u,v\in S$ and all possible additive weights $d(u),d(v)$ can be computed and represented in $\tilde O(r b^2)$ time and space. Then, given any additive weights $d(\cdot)$ to $S$, a representation of the Voronoi diagram w.r.t these weights can be constructed in $\tilde O(|S|)$ time. 
With this representation, for any site $u \in S$ we can query the maximum distance from $u$ to a vertex in $\Vor(u)$ in $\tilde O(|\partial \Vor(u)|) $ time. 
\end{theorem}

\section{Static Diameter}
\subsection{An  $n^{3+o(1)}/D^2$  Algorithm}\label{subsec:largeD}
In this subsection we prove \cref{thm:staticlargediameter}, stating that the diameter can be computed in $n^{3+o(1)}/D^2$ time on an unweighted undirected planar graph with diameter $D$. 
We first present a randomized $\tilde O(n^4/D^3)$ time algorithm, and then show how to improve it to $n^{3+o(1)}/D^2$. We then show how to derandomize both algorithms. 
We begin with two simple observations about the BFS levels when the diameter is $\geq D$. 

\begin{observation}\label{obsv:middle}
Let $s$ be any node in a graph of diameter $\geq D$. Then at least one out of the $D/2$ middle levels of the BFS tree rooted at $s$ has size $O(n/D)$.
\end{observation}

\begin{observation}\label{obsv:singleface}
Let $s$ be any node in $G$ and let $L_i$ be the set of nodes at level $i$ in the BFS tree rooted at $s$.
Let $G_i$ be the subgraph of $G$ that is induced by $\bigcup_{j \geq i} L_j$.
Then for each connected component $C$ of $G_i$ the nodes in $L_i \cap C$ lie on a single face.  
\end{observation}

\begin{proof}
To see that the vertices of $L_i \cap C$ all lie on the same face in $G_i$, consider the embedding of the component $C$ of $G_i$ inherited from the embedding of $G$. 
Viewing $C$ as a graph obtained from $G$ by deleting edges and vertices, one can start from any vertex of $L_i$ and follow a curve in the plane that only goes through deleted edges and vertices until reaching the root $s$ of the BFS tree. Hence all vertices of $L_i$ lie on a single face of $C$, and hence also of $G_i$.
\end{proof}

\para{A randomized algorithm.} 
We first compute in $O(n)$ time a 2-approximation (lower bound) $\tilde{D}$ of $D$ by computing a BFS tree and choosing $\tilde{D}$ to be  the furthest root-to-leaf distance. Then, we repeat the following procedure $\theta(n\log{n}/\tilde{D})$ times, and return the largest distance found:
 \begin{enumerate}
\item Randomly sample a source $s$, compute its BFS tree. Let $D'$ be the depth of this tree. Note that $D\ge D' \ge D/2$. Let $S=L_i$ be the set of nodes at level $i$ satisfying both $D'/4<i<3D'/4$ and  $|S|=O(n/D')=O(n/D)$. By \cref{obsv:middle}, such a set exists. Let $G_i$ be the subgraph of $G$ induced by $\bigcup_{j \geq i} L_j$.  
\item Compute $d(v,b)$ for all $v \in G$ and all $b \in S$. 
\item For each connected component $C$ of $G_i$:
\begin{enumerate}
\item Compute all bisectors in $C$ of  sites $C \cap S$ (that lie on a single face by \cref{obsv:singleface}). 
\item For each node $v$ in $G \setminus G_i$, compute the VD of $C$ w.r.t  the additive weights $d(v,b)$, and compute the distance from $v$ to its furthest vertex in every Voronoi cell of the VD. 
\end{enumerate}
\end{enumerate}

\para{Running time.} 
The first step takes $O(n)$ time by computing and traversing the BFS tree of $s$. The second step takes $O(n^2/D)$ time by doing a BFS from each vertex of $S$ in $O(n)$ time. 
The most expensive step is 3a. By \cref{thm:vor}, all bisectors of a  connected component $C$ can be computed in $\tilde O(|C| \cdot |C \cap S|^2)$ time. Over all connected components, this sums up to $\tilde O(n\cdot (n/D)^2)$ (since the $C$'s are disjoint and sum up to $n$, and the $C \cap S$ are disjoint and sum up to $O(n/D)$).   
Finally, in step 3b, for each vertex $v$, computing $v$'s VD and furthest vertex in every cell takes  $\tilde O(|C \cap S|)$ time by \cref{thm:vor}. Over all connected components, this sums up to $\tilde O(n/D)$, and thus over all vertices $v$ to $\tilde O(n^2/D)$. 
The total running time of the entire procedure is thus $\tilde O(n\cdot (n/D)^2)$, and since we repeat the procedure $\tilde O(n/D)$ times we get $\tilde O(n^4/D^3)$. 

\para{Correctness.} 
It remains to prove that the distance we return is indeed the diameter with high probability.
Let $x,y$ be the two endpoints of the diameter (i.e. $D=d(x,y)$). Then, the probability that a random source $s$ satisfies $d(s,x)\leq D'/4$ and $d(s,y)\geq 3D'/4$ is at least $D'/4n$ (because this happens if $s$ is one of the first $D'/4$ nodes on the path from $x$ to $y$).
Therefore, this happens with high probability for at least one of the sources $s$ that we choose. 
For this $s$, we will have that $x \in G\setminus G_i$ while $y \in G_i$ (it is impossible that $y \in G\setminus G_i$ because then an $x$-to-$y$ path through $s$ would be shorter than $D$), and then the largest distance that we find is guaranteed to be $d(x,y)$.


\para{Derandomization.} Observe that to derandomize the algorithm, it suffices to replace the sampling of sources with a (deterministic) selection of a set of sources $\mathcal S$ of size $O(n/D)$ such that a diameter endpoint $x$ is at distance $\le D'/4$ from at least one source $s\in \mathcal S$. 

To construct $\mathcal S$, pick an  arbitrary source $s$ and compute it's BFS tree $T$ of depth $D' \le D$. Find a level $L_i$ that has only $O(n/D')=O(n/D)$ nodes and $0.4 D' \le i \le 0.5 D'$. Similarly, find a level $L_j$ that has only $O(n/D)$ nodes and $0.8 D' \le j \le 0.9 D'$. The set of sources is then $\mathcal S=\{s\}\cup L_i \cup L_j$. It is easy to verify that every vertex $v$ in the graph has an ancestor or a descendant in $T$ that belongs to $\mathcal S$ and is at distance at most $D'/4 \le D/4$ from $v$.

\para{A faster algorithm.} 
Next, we improve the running time to $n^{3+o(1)}/D^2$. Again, we will start with a randomized algorithm and then derandomize. 
Let $B_\rho(v)$ denote the ball with radius $\rho$ around vertex $v$. 
Recall that our goal is to sample w.h.p. a vertex $s$ in $B_{\tilde D/4}(x)$ (without knowing $x$), where $x$ is a diameter endpoint. 

Let $\rho = \tilde D/4$. In order to sample a vertex $s$ in $B_{\rho}(x)$ w.h.p., it suffices to randomly sample a set of $\tilde O(n/|B_{\rho}(x)|)$ vertices (rather than sampling $\tilde O(n/\rho)$ vertices as in the approach above). Then, for each sampled vertex $s$, we can find a level  $L_i$ in the BFS tree of $s$ with $\rho< i \leq 2\rho$ s.t. $|L_i|<|B_{2\rho}(s)|/\rho$ (rather than $n/\rho$ as in the approach above).
Then, executing the approach above (i.e., executing steps 2--3 of the $\tilde O(n^4/D^3)$ algorithm above) for a specific $s$ would take time 
$\tilde O(n(|B_{2\rho}(s)|/\rho)^2)$ to compute all bisectors, $\tilde O(n|B_{2\rho}(s)|/\rho)$ to compute all additive weights, and $\tilde O(|B_{2\rho}(s)|^2/\rho)$ to construct the Voronoi diagrams for all vertices above level $i$. 
We see that if $|B_{\rho}(x)|$ is large then we gain because we have to sample fewer vertices, and if $|B_{2\rho}(s)|$ is small then we gain because the amount of work for each sampled vertex decreases. 

For this approach to work, we need to (1) estimate $|B_{\rho}(x)|$, and (2) relate $|B_{\rho}(x)|$ and $|B_{2\rho}(s)|$.
To address (1), we simply estimate $|B_{\rho}(x)|$ by enumerating all powers of two $2^k$ for $0 \leq k \leq \log n$. 
To address (2), note that $|B_{\rho}(x)| < |B_{2\rho}(s)| < |B_{3\rho}(x)|$, and that 
there must exist a $j \in \{1,2, \dots, \sqrt{\log_3 n}\}$ s.t. $|B_{3\rho 3^{-j}}(x)|/|B_{\rho 3^{-j}}(x)|<3^{\sqrt {\log_3 \!n}} = n^{o(1)}$ (if not, $|B_\rho(x)| > n$, a contradiction). 

The algorithm is therefore: For each $1\leq j \leq \sqrt{\log_3 \!n}$, let $\rho_j = 3^{-j} \rho$. 
For each $0 \leq k < \log n$ we sample $(n\log n)/2^k$ vertices $s$ (reflecting our assumption that $B_{\rho_j}(x) \leq 2^k$).
For each sampled vertex $s$, if $|B_{2\rho_j}(s)| > 2^k 3^{\sqrt {\log_3 \!n}}$, then, since $|B_{\rho}(x)| < |B_{2\rho}(s)| < |B_{3\rho}(x)|$, it must be that $s \notin B_{\rho_j}(x)$ or $|B_{\rho_j}(x)| > 2^k$ or $|B_{\rho_{j-1}}(x)|/|B_{\rho_j}(x)| > 3^{\sqrt {\log_3 \!n}}$ (the disjunction is not exclusive). 
Hence, in this case we discard $s$ and move on to the next sampled vertex. 
Otherwise, $|B_{2\rho_j}(s)| \leq 2^k 3^{\sqrt {\log_3 \!n}}$, and we can find a level $L_i$ with $\rho_j < i <2\rho_j$ in the BFS tree rooted at $s$ s.t. $|L_i| < 2^k 3^{\sqrt {\log_3 \!n}}/\rho_j$, and continue as in steps 2--3 from the previous algorithm. 
The overall running time is $$\sum_{j=0}^{\sqrt{\log_3 \!n}} \sum_{k=0}^{\log n} \tilde O\left(\frac{n}{2^k}\left(n(2^k3^{\sqrt{\log_3 \!n}}/\rho_j)^2 + n2^k3^{\sqrt{\log_3 \!n}}/\rho_j + (2^k3^{\sqrt{\log_3 \!n}})^2/\rho_j\right)\right) = n^{3+o(1)}/D^2.$$

To argue correctness, note that for $j$ such that $|B_{\rho_{j-1}}(x)|/|B_{\rho_j}(x)| \leq 3^{\sqrt {\log_3 \!n}}$ and $k$ such that $2^{k-1} \leq |B_{\rho_j}(x)| \leq 2^k$, sampling $(n\log n)/2^k$ vertices will yield with high probability a vertex $s \in B_{\rho_j}(x)$, and this $s$ will not be discarded. This $s$ satisfies $d(s,x)\leq \rho_j$ and $d(s,y)\geq 2\rho_j$, so the largest distance found for this $s$ is guaranteed to be $d(x,y)$ by the same argument as in the correctness of the slower algorithm.  

\para{Derandomization.} We use sparse neighborhood covers of Busch, Lafortune and Tirthapura \cite{BuschLT14} to derandomize the algorithm. 
A $\rho$-neighborhood cover $Z$ of a graph $G$ is a set of connected subgraphs called clusters, such that the union of all clusters is the vertex set of $G$ and such that for each node $v \in G$, there is some cluster $C \in Z$ that contains $B_\rho(v)$. 
The radius of a cover $Z$ is the maximum radius of a cluster in $Z$. 
The degree of a cover $Z$ is the maximum number of clusters that a node in $G$ is a part of. 
Busch et al. gave a deterministic $O(n \log n)$-time algorithm for computing, for any $\rho>0$ and any connected planar graph, a $\rho$-neighborhood cover of any connected planar graph with radius $16\rho$ and degree 18. 
See also~\cite{LeW21} for an $O(n)$ time algorithm.

To adjust the arguments we redefine $\rho_j = \rho33^{-j}$ for $j=1, \dots, \sqrt{\log_{33}(n)}$, and use the fact that for some $j$, $|B_{\rho_{j-1}}|/|B_{\rho_j}| < 33^{\sqrt{\log_{33}n}}$.
To avoid sampling in our algorithm, for each choice of $j,k$, we compute a $\rho_j$-neighborhood cover $Z$. We pick an arbitrary vertex $s$ from each cluster $C$ of $Z$ such that $|C|>2^k$. Since the degree of $Z$ is 18, the number vertices $s$ we choose is at most $18n/2^k$.

If $2^k<|B_{\rho_j}(x)|>2^{k+1}$ then the cluster $C$ containing $B_{\rho_j}(x)$ will have $|C|>2^k$ vertices, and we will choose a vertex $s\in C$.
Since the radius of $Z$ is $16\rho_j$, $d(s,x) \leq 16\rho_j$. 
If $|B_{17\rho_j}(s)| > 2^{k+1} 33^{\sqrt{\log_{33}n}}$, we discard $s$. 
Since $B_{17\rho_j}(s)$ is contained in $B_{33\rho_j}(x)=B_{\rho_{j-1}(x)}$, we are guaranteed that some $s$ will not be discarded.
For such $s$ we find a level $L_i$ with $16\rho_j < i < 17\rho_j$ in the BFS tree rooted at $s$ s.t. $|L_i|<2^{k+1} 33^{\sqrt{\log_{33}n}}/\rho_j$. 
The level of $x$ in the BFS tree is at most $16\rho_j$, and since $\rho_j < \rho < D/4$, the vertex $y$ such that $d(x,y)=D$ is at level greater than $i$ in the BFS tree. 
Hence, executing lines 2--3 of the procedure for the algorithm in section ~\ref{sec:replacement} will report the distance $D$.
The running time analysis is identical to that of the randomized version since we made sure that the number of vertices we choose in the derandomiztion is at most  some fixed constant times the number of sampled vertices in the randomized algorithm.

\subsection{An $\tilde O(n \cdot D^2)$  Algorithm}\label{sec:smallD}
In this subsection we prove \cref{thm:staticsmalldiameter}, stating that the diameter can be computed in $\tilde O(n \cdot D^2)$ time on an unweighted planar graph with diameter $D$.
We begin with some preliminaries on a recursive decomposition using shortest path separators. 

\para{Preliminaries.}
A \emph{shortest path
  separator} of a planar graph $G$ is an undirected cycle $C(G)$ consisting of a shortest $s$-to-$u$ path, a shortest $s$-to-$v$
path, and a single
edge $uv$, such that both the interior and exterior of the cycle 
consist of at most $2/3$ of the total number of the faces of $G$. 
Such a separator can be found in $O(n)$ time~\cite{LT79}. By recursively separating $G$ with shortest path separators (halting the recursion when we reach subgraphs of size $\le D$), we obtain the {\em decomposition tree} $\mathcal T$. The root of $\mathcal T$ corresponds to the entire graph $G$. A node corresponding to subgraph $P$ (we interchangeably refer to it as node $P$) has two children, whose subgraphs correspond to the interior and exterior of the separator $C(P)$. 

Observe that for every node $P \in \mathcal T$ the size of the shortest path separator $C(P)$ is $O(D)$. 
This is because $C(P)$ consists of two shortest paths, each of length at most $D$.
Moreover, the boundary of $P$ (vertices of $P$ that have incident edges to vertices not in $P$) is included in the union of all $C(P')$ where $P'$ is an ancestor of $P$, and is therefore of size $O(D\log n)$ and lies on $O(\log n)$ faces of $P$.
We compute the DDGs of every node (subgraph)  $P\in \mathcal T$ (i.e. copmute a data structure that can report in $\tilde O(1)$ time the distances in the graph $P$ between and pair of boundary vertices of $P$) using $O(\log n)$ executions of MSSP on $P$. This takes total $\tilde O(n)$ time over the entire $\mathcal T$. Now, given any vertex $v$ in the subgraph $P$, we can compute the distances in $G$ from $v$ to all boundary vertices of $P$ in $\tilde O(D)$ time using FR-Dijkstra. 
Namely, we initialize the $\tilde O(D)$ boundary vertices of $P$ to their distances from $v$ in the graph $P$ (via MSSP queries), and we run FR-Dijkstra on the union of the DDG of $P$ and the DDGs of all $P'$ where $P'$ is a sibling of some ancestor of $P$.

\para{The algorithm.}
For every non-leaf node $P \in \mathcal T$, we compute the furthest pair of vertices $u,v\in P$ where $u$ is internal to $C(P)$ and $v$ is external to $C(P)$. Observe that distances must be taken in the entire graph $G$ since the shortest $u$-to-$v$ path may venture out of $P$. To this end, we precompute all bisectors of every graph $P \in \mathcal T$, with the sites being the $\tilde O(D)$ boundary vertices of $P$. Using \cref{thm:vor}, this takes $\tilde{O}(|P| \cdot D^2)$ time (where $|P|$ denotes the size of the subgraph $P$), so over all $\mathcal T$ this takes $\tilde{O}(n \cdot D^2)$ time. (Observe that here we have used \cref{thm:vor} with the sites lying on $O(\log n)$ faces. As far as we know, in all prior uses of \cref{thm:vor} the sites lie on $O(1)$ faces).
Then, for every vertex $v\in P$, we compute the distances in $G$ from $v$ to all boundary vertices of $P$ using FR-Dijkstra in $\tilde O(D)$ time as explained above. We then use these distances as additive weights and apply \cref{thm:vor} to find the furthest vertex from $v$ in $P$. This also takes $\tilde O(D)$ time, so overall $\tilde O(n\cdot D)$.

We handle the leaf nodes $P \in \mathcal T$ explicitly (recall that $|P| \le D$). For each leaf node $P$ we compute the all-pairs shortest-paths (APSP) in $G$ between any two vertices $u,v\in P$. 
This is done by running Dijkstra's standard algorithm from every $v\in P$ on the graph $P$ where the boundary vertices of $P$ are initialized to their distances from $v$ in $G$ (that we have already computed as $v$'s  additive weights). This takes $\tilde O(D)$ time per $v$, so $\tilde O(D^2)$ time per $P$, and $\tilde O(D^2\cdot n/D) = \tilde O(nD)$ over all leaves $P$. 

\section{A Lower Bound on Dynamic Diameter}\label{sec:lower_bound}
In this section we prove \cref{thm:lower}. Namely, we give a conditional lower bound ruling out an amortized $O(n^{1-\eps})$ update time for maintaining the diameter of {\em weighted} planar graphs that undergo a sequence of edge-weight updates.

The proof is inspired by \cite{AD16}, however, there are quite a few changes since the reduction in \cite{AD16} is from APSP (not SETH), to dynamic distance oracles (not dynamic diameter), and rules out $O(n^{0.5-\eps})$ update time (not $O(n^{1-\eps})$). Our reduction is from the following problem, which is simply a recasting of the Orthogonal Vectors problem in the language of graphs.

\begin{definition}[Graph OV]
Given an undirected tripartite graph $G$ with parts $A,C,B$ where $|A|=|B|=n$ and the middle level  has size $|C|=O(\log{n})$, where all edges are in $A\times C$ and $C \times B$ decide if there exists a pair $a_i \in A, b_j \in B$ such that $d_G(a_i,b_j)>2$.
\end{definition}

It is known that solving this Graph OV problem in $O(n^{2-\eps})$ time refutes SETH \cite{Wil04,RV13}. Moreover, in the unbalanced version where $|A|=n^\alpha$ and $|B|=n^\beta$ for arbitrary constants $\alpha,\beta>0$ we know that an $O(n^{\alpha+\beta-\eps})$ time algorithm refutes SETH. 

\para{The structure of the reduction.}
Given an instance $G$ of the Graph OV problem, we construct a dynamic planar graph $H$. 
The graph $H$ is composed of two grids, a left grid and a right grid, each of dimension $|C|=O(\log n)$ by $|A|=n$. 
The columns of both grids are indexed by the nodes of $A$, such that the top node of the $i^{th}$ column in the left (resp. right) grid is called $a_i$ (resp.  $a_i'$).
The rows of the grids correspond to the nodes in $C$ such that the rightmost (resp. leftmost) node in the $k^{th}$ row of the left (resp. right) grid is called $c_k$ (resp. $c_k'$).
In both grids, all horizontal edges have weight $2 |C|$.
In the left grid, the vertical edges in column $i$ have weight $2i$ and in the right grid the vertical edges of column $i$ have weight $2(n-i)$. 
In the left grid, for every $i$ and $k$, if the edge $(a_i,c_k)$ exists in $G$, then we add a diagonal edge $e_k$ from vertex $(k-1,i)$ to vertex $(k,i+1)$ whose weight is $2i+2|C|-1$. We call such $e_k$ a \emph{shortcut} edge (as it is shorter by 1 compared to the alternative path composed of a vertical edge followed by a horizontal edge).  
The two grids are connected by $|C|$ edges:  for each $k$ we have an edge from $c_k$ to $c_k'$ of weight $2n |C| - 2 n k$. These $|C|$ edges are the only edges in $H$ whose weights will change throughout the reduction - all others will remain fixed.
We add a single node $x$ that is connected to all nodes in the top row of the left grid and all nodes of the top row in the right grid. We set the weight of every edge $(a_i,x)$ to be $i \cdot 4|C|$ and the weight of every edge $(x,a_j')$ to be $(n-j) \cdot 4|C|$.

\para{The dynamic updates.}
 After constructing the initial graph $H$ as above, for every $j=1,\dots,n$ we obtain a graph $H_j$ 
 by applying the following updates to $H$: for every $k = 1,\dots,|C|$ if the edge $(c_k,b_j)$ exists in $G$ then decrease by 1 the weight of the edge $(c_k,c_k')$ in $H$ (we refer to such edge $(c_k,c_k')$ as a {\em decreased} edge). The following main lemma shows that the diameter of $H_j$ reveals whether or not there exists an $a_i \in A$ such that $d_G(a_i,b_j)>2$. 
 Note, crucially, that we can generate all graphs $H_1,\ldots,H_n$ in sequence using only $O(n \log n)$ updates since $H_i$ differ from $H_{i-1}$ by only $O(\log n)$ edge weights.
Under SETH, we cannot maintain the diameter throughout this sequence  in $O(n^{2-\eps})$ time. Therefore, each update cannot be done in $O(n^{1-\eps})$ amortized time, thus proving \cref{thm:lower} for the fully-dynamic case.
To get a proof for the incremental case where edge weights only decrease we can do the following (the decremental case is symmetric).
Redefine the weight of the $O(\log n)$ edges so that they only decrease during the sequence: add $2(n-i)$ to their weight in $H_i$ so that their weight is the largest in $H_1$ and smallest in $H_n$.
Then, the sequence of graphs can be generated by only $O(n \log n)$ decrease-weight updates.
The diameter of $H_i$ increases by exactly $2(n-i)$ so the same analysis goes through.
For simplicity, we continue the proof in this section with the construction in the fully-dynamic case.

\begin{lemma}
For any $j$, the diameter of $H_j$ is larger than $4n|C|-2$ iff there exists $a_i \in A$ such that  $d_G(a_i,b_j)>2$.
\end{lemma}
In the remainder of this section we prove the above lemma. First observe that by our choice of edge-weights the diameter of $H_j$ correspond to some shortest $a_i$-to-$a_\ell'$ path. The following claim shows that in fact it is an $a_i$-to-$a_i'$ path.

\begin{claim}
For all $i \neq \ell$, $d_{H_j}(a_i,a_\ell') < 4n|C|-2$.
\end{claim} 
\begin{proof}
If $\ell>i$, then the path $a_i - x - a_\ell'$ consisting of two edges costs $(n-(\ell-i)) \cdot 4|C| < 4n|C|-2$. 
Otherwise $\ell<i$, then the $a_i$-to-$a_\ell'$ path that only uses horizontal edges costs $2|C|(n-i+\ell)+ 2|C| = 2n < 4n|C|-2$. 
\end{proof}

The following claim concludes the proof. 
\begin{claim}\label{claim:12}
For any $i$, $d_{H_j}(a_i,a_i') > 4n|C|-2$ iff $d_G(a_i,b_j)>2$.
\end{claim}
\begin{proof}
Observe that the path $a_i - x - a_i'$ consisting of two edges costs $4n|C|$. There may however be a shorter $a_i$-to-$a_i'$ path that passes through the grids.  
By our choice of edge weights (similarly to \cite{AD16}) such shortest path must start with an $a_i$-to-$c_k$ prefix (for some $k\le |C|$) in the left grid, then use the $(c_k,c_k')$ edge, then in the right grid do $i$ horizontal steps followed by $k$ vertical steps. Moreover, the $a_i$-to-$c_k$ prefix starts with $k-1$ vertical steps, then uses a shortcut edge $e_k$ if it exists (otherwise it does a horizontal step followed by a vertical step), and then it does $n-i-1$ horizontal steps until reaching $c_k$.   

Suppose first that there were no shortcut edges and no decreased edges at all.
The length of such $a_i$-to-$a_i'$ path would then be
$$ d_H(a_i,a_i') =  k \cdot 2i + (n-i) \cdot 2|C| + (2n |C| - 2 n k) + i\cdot 2|C|+ k \cdot 2(n-i) = 4n|C|. 
$$

Note that this length ($4n|C|$) is the same independent of $k$ and of $i$. Hence, the only way an $a_i$-to-$a_i'$ path can be shorter than $4n|C|$ is by using shortcut edges and decreased edges. However, it can use at most one shortcut edge $e_k$ and one decreased edge $(c_k,c_k')$. So its length is $4n|C|-2$ iff there exists a $k$ such that the shortcut $e_k$ exists (i.e., $(a_i,c_k) \in E(G)$) and the edge $(c_k,c_k')$ is decreased (i.e., $(c_k,b_j) \in E(G)$), and this is iff      $d_G(a_i,b_j)\leq 2$.
\end{proof}

\begin{remark}
By subdividing all edges, the above reduction implies that $O(n^{1/2-\eps})$ update time is impossible for maintaining the diameter of {\em unweighted} planar graphs.     
\end{remark}

\section{Decremental Voronoi diagrams and Replacement Diameter}\label{sec:replacement}


\para{Overview: A Step Toward Dynamic Voronoi Diagrams.}
The usefulness of Voronoi diagrams for diameter and distance reporting in static planar graphs make it natural to ask whether one can efficiently maintain some useful representation of Voronoi diagrams in the dynamic setting. 
This seems challenging because a change in a single edge or in a single additive weight can cause the entire Voronoi diagram to completely change. For example, decreasing the weight of a single edge in the Voronoi cell $\Vor(b)$ of some site $b$ may cause an expansion of $\Vor(b)$ on the expense of every other Voronoi cell, even cells that were not neighbors of $\Vor(b)$ before the change. The same is true for decreasing the additive weight of $b$.
Indeed, the few attempts to use Voronoi diagrams in dynamic planar graphs that we are aware of~\cite{faultyOracle,Charalampopoulos22}, all recompute the Voronoi diagrams from scratch upon every update. 

We make a small step towards dynamic Voronoi diagrams by developing a mechanism for updating Voronoi diagrams in the decremental setting. In our opinion, this is the most novel technical contribution of our work.
The deletion of an edge in some part of the graph only causes an increase in the additive weights of certain sites. 
When the additive weight of site $b$ increases, its Voronoi cell shrinks. Namely, some vertices that were in $\Vor(b)$ before the increase will belong to Voronoi cells of other sites after the increase. 
The crucial observation is that the only relevant sites in this process are $b$ and the sites of the neighboring cells to $\Vor(b)$ in the original Voronoi diagram. The time for the resulting update procedure is therefore proportional to the cell-degree of $b$, rather than to the total number of sites in the Voronoi diagram. Unfortunately, the cell-degree of $b$ may, in general, be as large as the number of sites. 
Nonetheless, this procedure turns out to be useful for the replacement diameter problem, where we can bound the number of times each site is affected by some edge deletion. 


\para{Representing Voronoi diagrams.}
Let $P$ be a planar graph with real edge-lengths.
Let $S$ be the set of vertices ({\em sites}) that lie on $\tilde O(1)$ faces, called holes.
Recall that every site $s\in S$ has an associated additive weight $d(s)$. 

Consider the Voronoi diagram of $P$ with sites $S$ and additive weights $d(s)$. Let $P^*$ be the planar dual of $P$. Let $\VD_0$ be the subgraph of $P^*$ consisting of the duals of edges $(u,v)$ of 
$P$ such that $u$ and $v$ are in different Voronoi cells. 
Let $\VD$ be the graph obtained from $\VD_0$ after eliminating all degree-2 vertices by repeatedly contracting any one of their incident edges. 
The vertices of $\VD$ are called {\em Voronoi vertices} and the edges of $\VD$ are called {\em Voronoi edges}. 
Observe that every Voronoi edge corresponds to a consecutive segment of some bisector between two sites. 
Note that $\VD$ may be disconnected, i.e., a planar map, and that the boundaries of faces of this planar map may be disconnected.
Each face of $\VD$ corresponds to a cell
$\Vor(s)$ of some site $s \in S$. Hence there are at most $|S|$ faces in $\VD$. 
It is shown in~\cite{DiameterSODA18} that the total 
number of edges, vertices, and faces of $\VD$ is $O(|S|)$.
In what follows, when we say we compute a Voronoi diagram $\VD$, we mean we use the algorithm in Theorem~\ref{thm:vor}, which computes a representation of the planar map $\VD$ defined above. Each Voronoi edge of $\VD$ corresponds to a segment of a bisector. 

\subsection{Maintaining Voronoi diagrams while additive weights increase.}\label{sec:shrinkingcells}
Consider an increase in the additive weight of a set $B\subseteq S$ of sites. Such an increase can only change the shortest path (w.r.t. additive weights) of vertices $v$ in the Voronoi cells of sites in $B$. 
Either the shortest path to such $v$ remains the same but its length increases by the change in the additive weight of the site, or $v$ becomes part of a Voronoi cell of a different site. In the latter case, since each shortest path is entirely contained in a single Voronoi cell, planarity dictates that the new site must be a neighbor of a site in $B$. We define the set $N(B)$ of neighbors of the sites in $B$ as the set of sites that are either in $B$ or sites whose Voronoi cells are adjacent to the Voronoi cells of sites in $B$.
Note that $|N(B)| = O(\sum_{b \in B}\celldegree(b))$. It follows from the discussion above that the only sites whose Voronoi cells might change as a result of such an increase are those in $N(B)$.

To compute the new Voronoi diagram we first compute the Voronoi diagram of $P$ with just the sites $N(B)$. By \cref{thm:vor}, this takes $O(\sum_{b \in B}\celldegree(b))$ time. Let $\VD'$ denote this Voronoi diagram, and let $\VD$ denote the Voronoi diagram of $P$ before the change. We stress that the additive weights of $\VD'$ are the ones after the increase, and the additive weights of $\VD$ are the ones before the increase. To obtain the Voronoi diagram of $P$ after the change, we ``glue'' together parts of $\VD'$ and $\VD$ as follows. See~\cref{fig:VD} for an illustration.

\begin{figure}[t!!]
  \begin{center}
\includegraphics[scale=0.20]{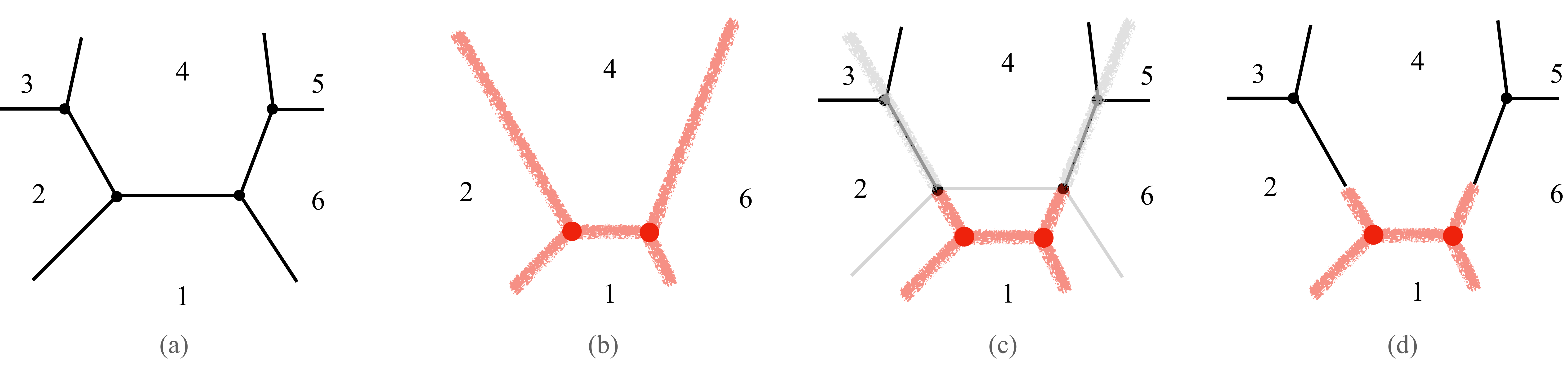} 
  \caption{An illustration of the process of computing the Voronoi diagram of a piece with 6 sites when the additive weight of site 1 is increased. 
(a) $\VD$, the Voronoi diagram of all 6 sites before the weight increase. 
(b) $\VD'$, the Voronoi diagram of just the increased site (1) and its neighbors (2, 4, 6), after the increase. 
(c) $\VD$ and $\VD'$ superimposed, with the edges deleted from $\VD$, and from $\VD'$ in grey. Observe that all segments of bisectors between cells of the neighbors (2,4,6) that appear in $\VD$ also appear in $\VD'$.
(d) The glued Voronoi diagram.     \label{fig:VD} 
  }
   \end{center}
\end{figure}

Recall that $\VD$ is a (possibly disconnected) planar map whose edges correspond to segments of bisectors of pairs of sites of $\VD$. The endpoints of these segments are Voronoi vertices of $\VD$. 
We start by deleting from $\VD$ all the Voronoi edges corresponding to bisectors involving at least one site of $B$. 
For every Voronoi vertex $v$ incident to a Voronoi cell of a site in $B$, if all the Voronoi edges incident to $v$ were deleted, then we delete $v$ as well. 
Let $\mathcal E$ denote the set of Voronoi edges $e$ of $\VD$ such that $e$ is incident to some Voronoi vertex $v$, $e$ was not deleted, but the preceding or following Voronoi edge of $e$ in the cyclic order of edges around $v$ was deleted.
Every Voronoi edge $e \in \mathcal E$ corresponds to a segment $\beta$ of a bisector between two sites $s_1,s_2 \in N(B)\setminus B$. 
Since the additive weights of these sites are unchanged, the segment $\beta$ must be represented by a Voronoi edge $e'$ of $\VD'$. 
Note that $\beta$ may be a sub-segment of the bisector segment $\beta'$ corresponding to $e'$. Also note that it is easy to identify $e'$ with $e$ during the computation of $\VD'$ with no asymptotic time overhead.\footnote{This can be done, for example, by augmenting the binary search tree representation of segments of bisectors used in the construction algorithm (cf.~\cite{DiameterSODA18}) with a boolean flag marking edges in $\mathcal E$. Then we can go over all Voronoi edges of $\VD'$ and list for each one the corresponding marked edge $e \in \mathcal E$, if such an edge exists in the segment of the bisector corresponding to that Voronoi edge.} 
For each Voronoi edge $e \in \mathcal E$ (of $\VD$), we split its corresponding Voronoi edge $e'$ (of $\VD'$) into two edges $e'_1,e'_2$ by breaking $\beta'$ into two sub-segments at $v$. Suppose $e'_2$ is the one whose corresponding bisector segment contains $\beta$. 
Note that if $v$ is an endpoint of $e'$ (i.e., if $v$ is a Voronoi vertex of $\VD'$ as well), then $e'_1$ is a trivial empty segment of the $s_1$-$s_2$ bisector. 
We delete $e'_2$ from $\VD'$, and merge the Voronoi edges $e$ of $\VD$ and $e'_1$ of $\VD'$ into a single Voronoi edge whose corresponding segment is the concatenation of the segment $\beta$ of $e$ and the segment of $e'_1$.

Doing so for the edges $e \in \mathcal E$ effectively ``glues'' the relevant portion of $\VD'$ into $\VD$, replacing the portion of $\VD$ that we had deleted. 
The algorithm of~\cite{DiameterSODA18} for constructing Voronoi diagrams from precomputed bisectors performs similar stitching and glueing operations, and the data structures used to represent Voronoi diagrams and bisectors support all the necessary operations in $\tilde O(1)$ time per operation.
Hence, the time complexity of this entire procedure is proportional (up to logarithmic factors) to the number of Voronoi vertices of the Voronoi cells of the sites in $B$, which is $O(\sum_{b \in B}\celldegree(b))$.


\subsection{Replacement Diameter}
We now describe how to use the new algorithm for maintaining Voronoi diagrams under additive weight decreases to get a faster algorithm for replacement diameter.
The algorithm starts by computing a complete recursive decomposition tree $\mathcal T$ of the graph $G$. For every node (piece) in $\mathcal T$ (corresponding to a subgraph of $G$) we compute all its bisectors. 
This takes $\tilde O(n^2)$ time over all $\mathcal T$ using \cref{thm:vor}. 
Then, for every vertex $s\in V$ we compute the BFS tree $T_s$ of $s$ in $G$ and compute the {\em fault-tolerant single source distance oracle} of Baswana et. al.~\cite{Baswana} for $s$ in $G$. 
This oracle is constructed in $\tilde O(n)$ time from $G$, and can report in $\tilde O(1)$ time the $s$-to-$t$ distance in the graph  $G^{e}=(V, E \setminus \{e\})$ for any $s,t\in V$ and any $e\in E$. Overall, this also takes $\tilde O(n^2)$ time. 
For each piece $P \in \mathcal T$, for each boundary vertex $b \in \partial P$ we create the {\em induced tree}  $T^P_b$ from $T_b$ by marking all vertices of $P$ and all their lowest common ancestors, and contracting any edge whose endpoints are not marked. 
The resulting $T^P_b$ has size $O(|P|)$. For each edge $e$ of $G$ that was contracted in the process we store the edge of $T^P_b$ into which $e$ was contracted. 
Since the total number of boundary nodes and piece sizes over all pieces of $\mathcal T$ is $\tilde O(n)$, the total time to construct all these induced trees is $\tilde O(n^2)$.
For each piece $P\in \mathcal T$, let $P'$ be the sibling of $P$ in $\mathcal T$. Let $b_1, b_2, \dots$ be the vertices of $\partial P'$ in some arbitrary order. For each vertex $s \in P$ we compute the additivley weighted Voronoi diagram of $s$ w.r.t $P'$ with sites $\{b_i\}$ and additive weights $d(s,b_i)$. We also store for $s$ a binary search tree (BST) over $b_1, b_2, \dots$, where the node $i$ in the tree stores the distance from $s$ to the furthest vertex in $\Vor(b_i)$. This takes total $\tilde O(n\sqrt n)$ time over all $P \in \mathcal{T}$ and all $s\in P$.
For each piece $P$ with vertices $v_1, v_2, \dots$ in arbitrary order, we store a BST over $\{v_i\}$, where node $i$ stores the furthest vertex from $v_i$ in $P'$. This vertex can be found in $\tilde O(1)$ time for each $v_i$ by querying the maximum distance stored in the BST of $v_i$.

For every edge $e \in E$, we need to compute the furthest pair of vertices in the graph $G^{e}=(V, E \setminus \{e\})$. 
For an edge $e\in E$ and two vertices $u,v\in G$, we say that the pair $u,v$ is {\em affected} in $G^{e}$ if $e$ lies on the root-to-$v$ path in $T_u$.
The main idea is to use the fact that a specific pair $u,v$ is affected in at most $D$ (rather than $n$) graphs $G^{e}$ (since the shortest $u$-to-$v$ path in $G$ has at most $D$ edges). 

For every affected pair $(u,v)$ there is some pair of sibling pieces $(P,P')$ s.t. $u\in P$ and $v\in P'$. Our strategy is to go over pairs of sibling pieces $(P,P')$ in $\mathcal T$, and handle all affected pairs for each $(P,P')$ together as follows. 
Assume w.l.o.g. that $e \notin P'$. For each $b \in \partial P'$, we enumerate in $T^P_{b}$ all the decendant vertices of the edge of $T^P_b$ into which $e$ was contracted (this may be an empty set if $e \notin T_b$). This way we identify all the affected pairs of the form $(u,b)$, where $u\in P$ and $b\in \partial P'$. 
We query the Baswana et al. oracle for the $u$-to-$b$ distance in $G^e$ for each such affected pair.
For each $u \in P$, let $B$ be the set of boundary vertices $b$ such that $(u,b)$ is an affected pair.
For each vertex $u \in P$ with $|B| \geq 1$, we update the Voronoi diagram of $u$ w.r.t. $P'$ using the procedure Decremental-VD, which is described in subsection~\ref{sec:shrinkingcells}. This procedure updates the VD (and the furthest vertex from each site) w.r.t the new additive weights in time $\sum_{b \in B} \celldegree(b)$ where $\celldegree$ is the number of Voronoi cells that are adjacent\footnote{Two cells are adjacent if there exists an edge $e$ of $G$ with one endpoint in each cell.} to the cell $\Vor(b)$ in the original VD (i.e. before the deletion of $e$). 
Using the updated VD, we update the node corresponding to every $b\in B$ in the BST of $u$ with the new furthest vertex in $\Vor(b)$. Let $d$ be the maximum distance stored in the entire BST of $u$. We update the node corresponding to $u$ in the BST of $P$ with the value $d$.
After handling all $u \in P$ with $|B| \geq 1$ in this way, the maximum value stored in the entire BST of $P$ is the maximum distance in $G^e$ between any pair of vertices $(u,v)$ with $u\in P$ and $v \in P'$.
Taking the maximum over all pairs of siblings $(P,P') \in \mathcal T$ gives the diameter of $G^e$.

The total running time for computing the furthest pair for the siblings $(P,P')$ is analyzed as follows. 
The bottleneck is the time to update the VDs.
Every time a pair $u,b$ (where $u \in P$ and $b \in \partial P'$) is affected we spend $\tilde O(\celldegree(b))$ time updating the VD of $u$. 
Since each pair is affected by the deletion of at most $D$ edges, the total time invested in updating VDs for $(P,P')$ is bounded by $\sum_{u\in P,b\in \partial P'} D \cdot  \celldegree(b) = |P| D \sum_{b} \celldegree(b)$, which is 
$\tilde O(|P| D \cdot |\partial P'|)$, since the sum of $\celldegree$s of the cells in a VD is order of the number of sites of the VD.
Summing over all pairs of sibling pieces we get that the total time is $\sum_{(P,P')\in \mathcal T} \tilde O(|P| D \cdot |\partial P'|) = \tilde O(n^{1.5} D)$. 
Hence, including the preprocessing, the total time for the entire replacement diameter algorithm is $\tilde O(n^2 + n^{1.5} D)$.

We note that when $D \geq n^{5/6}$, it is better to naively use the static $n^{3+o(1)}/D^2$-time algorithm from section~\ref{subsec:largeD} for each edge failure. 
Hence, replacement diameter can be solved in $\min(n^{3+o(1)}/D^2,\tilde O(n^2 + n^{1.5} D)) = n^{7/3 + o(1)}$ time.

\section{Incremental Diameter}

In this section we prove \cref{thm:incremental}. Namely, we present a general reduction showing how to solve the diameter problem efficiently in \emph{incremental} graphs given two components: (1) a distance oracle for incremental graphs, and (2) a diameter algorithm for static graphs that is relatively fast when the diameter is large.
Plugging in the incremental distance oracle of Das et al. \cite{DasGW22} and the static algorithm of Section~\ref{subsec:largeD} we obtain an algorithm with total time $n^{7/3+o(1)}$ which improves over the naive bound of $\tilde{O}(n^{8/3})$.
The new algorithm of this section comes closer to the $n^{2-o(1)}$ lower bound of Section~\ref{sec:lower_bound} for weighted graphs (the best lower bound for unweighted graphs is $n^{1.5-o(1)}$).

The rest of this section is dedicated to proving this theorem. We begin by presenting the general reduction (that does not assume planarity nor unweighted edges) and then explain how it can be combined with existing algorithms for planar graphs to obtain the theorem.

\para{A reduction from diameter to $s,t$-shortest path.}

In an incremental graph, the diameter decreases with time, starting from some $D\leq n$ (otherwise the graph is not connected and it is easy to check this efficiently) and ending at some $D\geq 1$. 
The idea for the reduction is simple: we would like to recompute the diameter only when it decreases, and not after each of the $n$ updates. 
While it is true that the diameter could decrease $\Omega(n)$ times, from $n$ to $1$, the point is that re-computation is efficient when the diameter is large (due to the $n^{3+o(1)}/D^2$ algorithm of Section~\ref{subsec:largeD}) and then only $O(D)$ of the re-computations will happen when the diameter is smaller than $D$. 

Our incremental algorithm works as follows:
\begin{itemize}
\item \textbf{Step 1 - sample a new diameter pair:} Let $P = \{ (s,t) \mid d(s,t)= \Delta(G)\}$ be the set of pairs that realize the current diameter $\Delta(G)$. Sample a pair $(s',t')$ from $P$ uniformly at random (or from some distribution in which every pair is sampled with probability at most $O(1/|P|)$).

\item \textbf{Step 2 - monitor the distance of the sampled pair:} Using an incremental distance oracle, monitor the distance between $s'$ and $t'$ throughout the sequence of edge insertions. Do nothing (except querying the oracle) as long as $d(s',t')$ does not decrease; in which case it is still the correct diameter of the graph and can be output whenever there is a query. If a new edge causes $d(s',t')$ to decrease, go back to Step 1.
\end{itemize}

Each of the two steps involves one of the two ingredients in our reduction. Step 2 utilizes an incremental distance oracle, while Step 1 uses a static diameter algorithm \emph{that can also sample a diameter pair}. At the end of this section we  give a general reduction from the latter approximate sampling problem to the problem of finding the largest distance from each node in the graph (i.e. computing all eccentricities). Alternatively, one could notice that the diameter algorithms we will employ in Step 1 (and many other natural diameter algorithms) can be modified to also sample a diameter pair uniformly at random. 

\para{Running time.}
Let us first bound the number of times we go to Step 1, which is the most costly step since it involves a static diameter computation. Step 2 is actually very cheap since we only perform one update and one query to an incremental distance oracle.

\begin{claim}\label{claim:15}
For any (non adaptive) sequence of edge insertions that does not decrease the diameter of the graph, the expected number of times our algorithm samples a diameter pair (i.e. goes to Step 1) is $O(\log n)$.
\end{claim}

\begin{proof}
Let us first analyze the idealistic case in which we manage to sample truly uniformly in Step 1, and then point out that the same analysis essentially goes through when we sample almost uniformly.

Each new edge $e$ decreases the distance for a subset of pairs $X_e \subseteq P$. Since the special pair $(s',t')$ is completely unknown to the adversary who is choosing the sequence of edge insertions, the probability that $e$ causes the algorithm to go to Step 1 is exactly $|X_e|/|P|$ and in that case the new set of ``diameter pairs'' becomes $P \setminus X_e$. Therefore, the expected number of times we sample can be upper bounded by:
$f(|P|) \leq \max_{0\leq x \leq |P|} x/|P| + f(|P|-x) = O(\log {|P|})$.

If the sampling in Step 1 is only approximately uniform, but still satisfies that a pair is chosen with probability at most $O(1/|P|)$ then the same analysis above goes through, up to an additional $O(1)$ factor.
\end{proof}

Let $T^{Diam}(n,D)$ denote the running time of a static diameter algorithm that samples a diameter pair as in Step 1, when the diameter of the graph is $D$. 
Over all the $O(n)$ edge insertions, the total expected running time of Step 1 is therefore at most $ \sum_{D=1}^n \log n \cdot T^{Diam}(n,D)$.

To obtain our claimed upper bound of $n^{7/3+o(1)}$ we will use two diameter algorithms inside this reduction: the $T^{Diam}(n,D) =n^{3+o(1)}/D^2$ algorithm from Section~\ref{subsec:largeD} (for large $D$) and the $T^{Diam}(n,D) =\tilde{O}(n^{5/3})$ algorithm \cite{DiameterSODA18} (for small $D$). 
(By the reduction in Section~\ref{subsec:sampling}, these algorithms can also sample an approximately uniform pair as required by Step 1).
The total expected time becomes:
$$ \sum_{D=1}^n \log n \cdot T^{Diam}(n,D) = \tilde{O}\left( \sum_{D=1}^{n^{2/3}}  n^{5/3} +   \sum_{D=n^{2/3}}^n  n^{3+o(1)}/D^2 \right) = n^{7/3+o(1)},$$
because $\sum_{D=n^{2/3}}^n  n^{3+o(1)}/D^2 \leq \sum_{i= \log_2{n^{2/3}}}^{\log_2{n}} 2^{i+1} \cdot n^{3+o(1)}/(2^i)^2 \leq \frac{n^{3+o(1)}}{n^{2/3}} \cdot 2\log{n}$.
The additional time of Step 2 is at most $n \cdot \sqrt{n}$ using the incremental distance oracle of Das et al. \cite{DasGW22} that has $O(\sqrt{n})$ time per update and query.

\subsection{Sampling a Diameter Pair}
\label{subsec:sampling}
In this section we show the final piece of the incremental diameter algorithm. Namely, a way to adapt the aforementioned static diameter algorithms so that they sample a diameter pair approximately uniformly.

A first attempt, that does not quite work, is to add a random ``perturbation'' $p_e \in (0,\eps)$ to the weight of each edge $e$, where $\eps < 1/D$, and then argue that the (probably unique) pair realizing the diameter in the new graph is a uniformly random pair in $P = \{ (s,t) \mid d(s,t)=\Delta(G) \}$.
Note that the perturbations increase the distance between all pairs by $<1$ and therefore non-diameter-pairs cannot become diameter pairs.
One issue, however, is that pairs with many paths of length $\Delta(G)$ between them are more likely to be chosen than pairs with few such paths.
A second attempt that resolves this issue is to add a perturbation to the nodes (e.g. by appending a private leaf to each node with a random weight on the new edge). 
This idea is closer to the actual solution but it still has an issue of correlations: a node that participates in many pairs might be sampled less frequently than a node that participates in few pairs.
Therefore, we must take this difference into account when assigning the weights. 

Making these ideas go through is a bit complicated. Fortunately, there is an elegant reduction from our setting to the \emph{bipartite independet set} query model introduced by Beame et al. \cite{beame2020edge} and then use existing results on this model \cite{DellLM22,bhattachrya2022faster,addanki2022non} in a black-box way.

\begin{theorem}
\label{thm:sampling}
There is an algorithm that samples a pair in $P=\{(s,t)\mid d(s,t)=\Delta(G)\}$ where each pair is sampled with probability at most $O(1/|P|)$ and runs in time $(\min(\tilde{O}(n^{5/3}),n^{3+o(1)}/D^2))$ on unweighted planar graphs of diameter $D$.
\end{theorem}

The main lemma towards proving the theorem is the following.

\begin{lemma}
By making $\log^{O(1)} n$ calls to an algorithm that returns all eccentricities we can sample a pair in $P=\{(s,t)\mid d(s,t)=\Delta(G)\}$ where each pair is sampled with probability at most $O(1/|P|)$
\end{lemma}

\begin{proof}

Consider an implicit graph $H$ in which there is an edge between two nodes $s,t$ iff they are a diameter pair in $G$ (i.e., $(s,t) \in P$).
Our goal is to sample an edge from $H$ approximately uniformly.
This can be achieved \cite{DellLM22,bhattachrya2022faster,addanki2022non} by making a polylogarithmic number of queries to an oracle that, given two subsets $L,R \subseteq V(H)$, decides whether there is any edge in $L \times R \cap E(H)$. This is called a bipartite independent set oracle in the literature, following Beame et al. \cite{beame2020edge}.
Thus, all we have to do is show that such a query can be supported in the time of a call to an algorithm that computes all eccentricities in the graph.

First, we precompute the diameter $\Delta(G)$ of $G$.
Then, given a query $L,R \subseteq V$ we construct a graph $G'$ from $G$ as follows.
For each node $v \in R$ we add a new ``leaf'' node $l_v$ and connect it with an edge (of weight $1$) to $v$.
Next, we compute the eccentricity of all nodes in $G'$. 
Finally, the answer to the query is yes if and only if there is a $u \in L$ such that the eccentricity of $u$ in $G'$ is $\Delta(G)+1$; this can be checked in $O(n)$ time.

The correctness of the answer follows from the observation that the eccentricity of any node $u$ in $G'$ is $\Delta(G)+1$ if and only if there is a node $v$ in $G$ such that (1) $d_G(u,v)=\Delta(G)$ and (2) a new leaf node $l_v$ was appended to $v$.
This implies that (1) $(u,v) \in P$ is a diameter pair in $G$, meaning that $(u,v) \in E(H)$, and that (2) $v \in L$.
Since we only check for $u \in R$ our answer that $L \times R \cap E(H)$ is non-empty is correct.
\end{proof}

To conclude the proof of Theorem~\ref{thm:sampling} we simply point out that both of the relevant diameter algorithms already compute the eccentricity of all nodes.



\bibliographystyle{abbrv}

\begin{thebibliography}{10}

\bibitem{AbboudCKP21}
A.~Abboud, K.~Censor{-}Hillel, S.~Khoury, and A.~Paz.
\newblock Smaller cuts, higher lower bounds.
\newblock {\em {ACM} Trans. Algorithms}, 17(4):30:1--30:40, 2021.

\bibitem{ACK20}
A.~Abboud, V.~Cohen{-}Addad, and P.~N. Klein.
\newblock New hardness results for planar graph problems in p and an algorithm
  for sparsest cut.
\newblock In {\em 52nd {STOC}}, pages 996--1009, 2020.

\bibitem{AD16}
A.~Abboud and S.~Dahlgaard.
\newblock Popular conjectures as a barrier for dynamic planar graph algorithms.
\newblock In {\em 57th {FOCS}}, pages 477--486, 2016.

\bibitem{AbboudVassilevskaFOCS14}
A.~Abboud and V.~V. Williams.
\newblock Popular conjectures imply strong lower bounds for dynamic problems.
\newblock In {\em FOCS}, pages 434--443, 2014.

\bibitem{AWW16}
A.~Abboud, V.~V. Williams, and J.~R. Wang.
\newblock Approximation and fixed parameter subquadratic algorithms for radius
  and diameter in sparse graphs.
\newblock In {\em SODA}, pages 377--391, 2016.

\bibitem{AbrahamCG12}
I.~Abraham, S.~Chechik, and C.~Gavoille.
\newblock Fully dynamic approximate distance oracles for planar graphs via
  forbidden-set distance labels.
\newblock In {\em 44th annual ACM symposium on Theory of computing}, pages
  1199--1218, 2012.

\bibitem{addanki2022non}
R.~Addanki, A.~McGregor, and C.~Musco.
\newblock Non-adaptive edge counting and sampling via bipartite independent set
  queries.
\newblock {\em arXiv preprint arXiv:2207.02817}, 2022.

\bibitem{SMAWK}
A.~Aggarwal, M.~M. Klawe, S.~Moran, P.~Shor, and R.~Wilber.
\newblock Geometric applications of a matrix-searching algorithm.
\newblock {\em Algorithmica}, 2(1):195--208, 1987.

\bibitem{AnconaHRWW19}
B.~Ancona, M.~Henzinger, L.~Roditty, V.~V. Williams, and N.~Wein.
\newblock Algorithms and hardness for diameter in dynamic graphs.
\newblock In {\em 46th {ICALP}}, volume 132, pages 13:1--13:14, 2019.

\bibitem{ArikatiCCDSZ96}
S.~R. Arikati, D.~Z. Chen, L.~P. Chew, G.~Das, M.~H.~M. Smid, and C.~D.
  Zaroliagis.
\newblock Planar spanners and approximate shortest path queries among obstacles
  in the plane.
\newblock In {\em 4th {ESA}}, volume 1136, pages 514--528, 1996.

\bibitem{BackursRSWW21}
A.~Backurs, L.~Roditty, G.~Segal, V.~V. Williams, and N.~Wein.
\newblock Toward tight approximation bounds for graph diameter and
  eccentricities.
\newblock {\em {SIAM} J. Comput.}, 50(4):1155--1199, 2021.

\bibitem{Baswana}
S.~Baswana, U.~Lath, and A.~S. Mehta.
\newblock Single source distance oracle for planar digraphs avoiding a failed
  node or link.
\newblock In {\em 23rd {SODA}}, pages 223--232, 2012.

\bibitem{beame2020edge}
P.~Beame, S.~Har-Peled, S.~N. Ramamoorthy, C.~Rashtchian, and M.~Sinha.
\newblock Edge estimation with independent set oracles.
\newblock {\em ACM Transactions on Algorithms (TALG)}, 16(4):1--27, 2020.

\bibitem{CactusGraphs}
B.~Ben{-}Moshe, B.~K. Bhattacharya, Q.~Shi, and A.~Tamir.
\newblock Efficient algorithms for center problems in cactus networks.
\newblock {\em Theor. Comput. Sci.}, 378(3):237--252, 2007.

\bibitem{BermanKasiviswanathan}
P.~Berman and S.~Kasiviswanathan.
\newblock Faster approximation of distances in graphs.
\newblock In {\em Proc. of the 10th International Workshop on Algorithms and
  Data Structures {(WADS)}}, pages 541--552, 2007.

\bibitem{bhattachrya2022faster}
A.~Bhattachrya, A.~Bishnu, A.~Ghosh, and G.~Mishra.
\newblock Faster counting and sampling algorithms using colorful decision
  oracle.
\newblock {\em Leibniz International Proceedings in Informatics (LIPIcs)}, 219,
  2022.

\bibitem{Bonnet22}
{\'{E}}.~Bonnet.
\newblock 4 vs 7 sparse undirected unweighted diameter is seth-hard at time
  \emph{n}\({}^{\mbox{4/3}}\).
\newblock {\em {ACM} Trans. Algorithms}, 18(2):11:1--11:14, 2022.

\bibitem{BorradailePW12}
G.~Borradaile, S.~Pettie, and C.~Wulff{-}Nilsen.
\newblock Connectivity oracles for planar graphs.
\newblock In {\em 13th {SWAT}}, volume 7357, pages 316--327, 2012.

\bibitem{BorradaileSW10}
G.~Borradaile, P.~Sankowski, and C.~Wulff{-}Nilsen.
\newblock Min st-cut oracle for planar graphs with near-linear preprocessing
  time.
\newblock In {\em 51th {FOCS}}, pages 601--610, 2010.

\bibitem{BuschLT14}
C.~Busch, R.~LaFortune, and S.~Tirthapura.
\newblock Sparse covers for planar graphs and graphs that exclude a fixed
  minor.
\newblock {\em Algorithmica}, 69(3):658--684, 2014.

\bibitem{Cabello12}
S.~Cabello.
\newblock Many distances in planar graphs.
\newblock {\em Algorithmica}, 62(1-2):361--381, 2012.

\bibitem{Cabello}
S.~Cabello.
\newblock Subquadratic algorithms for the diameter and the sum of pairwise
  distances in planar graphs.
\newblock In {\em 28'th {SODA}}, pages 2143--2152, 2017.

\bibitem{Chan12}
T.~M. Chan.
\newblock All-pairs shortest paths for unweighted undirected graphs in $o(mn)$
  time.
\newblock {\em {ACM} Trans. Algorithms}, 8(4):34:1--34:17, 2012.

\bibitem{ChanS19}
T.~M. Chan and D.~Skrepetos.
\newblock Faster approximate diameter and distance oracles in planar graphs.
\newblock {\em Algorithmica}, 81(8):3075--3098, 2019.

\bibitem{ChangKT22}
H.~Chang, R.~Krauthgamer, and Z.~Tan.
\newblock Almost-linear \emph{{\(\epsilon\)}}-emulators for planar graphs.
\newblock In {\em 54th {STOC}}, pages 1311--1324, 2022.

\bibitem{Charalampopoulos19}
P.~Charalampopoulos, P.~Gawrychowski, S.~Mozes, and O.~Weimann.
\newblock Almost optimal distance oracles for planar graphs.
\newblock In {\em 51st {STOC}}, pages 138--151, 2019.

\bibitem{Charalampopoulos22}
P.~Charalampopoulos and A.~Karczmarz.
\newblock Single-source shortest paths and strong connectivity in dynamic
  planar graphs.
\newblock {\em J. Comput. Syst. Sci.}, 124:97--111, 2022.

\bibitem{faultyOracle}
P.~Charalampopoulos, S.~Mozes, and B.~Tebeka.
\newblock Exact distance oracles for planar graphs with failing vertices.
\newblock In {\em 30th {SODA}}, pages 2110--2123, 2019.

\bibitem{ChechikLRSTW14}
S.~Chechik, D.~H. Larkin, L.~Roditty, G.~Schoenebeck, R.~E. Tarjan, and V.~V.
  Williams.
\newblock Better approximation algorithms for the graph diameter.
\newblock In {\em 25th {SODA}}, pages 1041--1052, 2014.

\bibitem{ChenX00}
D.~Z. Chen and J.~Xu.
\newblock Shortest path queries in planar graphs.
\newblock In {\em 32nd {STOC}}, pages 469--478, 2000.

\bibitem{Chordal}
V.~Chepoi and F.~F. Dragan.
\newblock A linear-time algorithm for finding a central vertex of a chordal
  graph.
\newblock In {\em ESA}, pages 159--170, 1994.

\bibitem{HyperbolicGeodesic}
V.~Chepoi, F.~F. Dragan, B.~Estellon, M.~Habib, and Y.~Vax{\`{e}}s.
\newblock Diameters, centers, and approximating trees of
  delta-hyperbolicgeodesic spaces and graphs.
\newblock In {\em SoCG}, pages 59--68, 2008.

\bibitem{PlaneTriangulations}
V.~Chepoi, F.~F. Dragan, and Y.~Vax{\`{e}}s.
\newblock Center and diameter problems in plane triangulations and
  quadrangulations.
\newblock In {\em SODA}, pages 346--355, 2002.

\bibitem{Euclidian}
K.~Clarkson and P.~Shor.
\newblock Applications of random sampling in computational geometry, ii.
\newblock {\em Discrete {\&} Computational Geometry}, 4(5):387--421, 1989.

\bibitem{Cohen-AddadDW17}
V.~Cohen{-}Addad, S.~Dahlgaard, and C.~Wulff{-}Nilsen.
\newblock Fast and compact exact distance oracle for planar graphs.
\newblock In {\em 58th {FOCS}}, pages 962--973, 2017.

\bibitem{DalirrooyfardLW21}
M.~Dalirrooyfard, R.~Li, and V.~V. Williams.
\newblock Hardness of approximate diameter: Now for undirected graphs.
\newblock In {\em 62nd {IEEE} Annual Symposium on Foundations of Computer
  Science, {FOCS} 2021, Denver, CO, USA, February 7-10, 2022}, pages
  1021--1032. {IEEE}, 2021.

\bibitem{DasGW22}
D.~Das, M.~P. Gutenberg, and C.~Wulff{-}Nilsen.
\newblock A near-optimal offline algorithm for dynamic all-pairs shortest paths
  in planar digraphs.
\newblock In {\em 33rd {SODA}}, pages 3482--3495, 2022.

\bibitem{DellLM22}
H.~Dell, J.~Lapinskas, and K.~Meeks.
\newblock Approximately counting and sampling small witnesses using a colorful
  decision oracle.
\newblock {\em {SIAM} J. Comput.}, 51(4):849--899, 2022.

\bibitem{Djidjev96}
H.~Djidjev.
\newblock Efficient algorithms for shortest path queries in planar digraphs.
\newblock In {\em 22nd {WG}}, volume 1197, pages 151--165, 1996.

\bibitem{DistanceHereditaryGraphs}
F.~F. Dragan and F.~Nicolai.
\newblock Lexbfs-orderings of distance-hereditary graphs with application to
  the diametral pair problem.
\newblock {\em Discrete Applied Mathematics}, 98(3):191--207, 2000.

\bibitem{Eppstein}
D.~Eppstein.
\newblock Subgraph isomorphism in planar graphs and related problems.
\newblock In {\em 6th {SODA}}, pages 632--640, 1995.

\bibitem{FR06}
J.~Fakcharoenphol and S.~Rao.
\newblock Planar graphs, negative weight edges, shortest paths, and near linear
  time.
\newblock {\em J. Comput. Syst. Sci.}, 72(5):868--889, 2006.

\bibitem{OuterplanarGraphs}
A.~M. Farley and A.~Proskurowski.
\newblock Computation of the center and diameter of outerplanar graphs.
\newblock {\em Discrete Applied Mathematics}, 2(3):185--191, 1980.

\bibitem{F87}
G.~N. Frederickson.
\newblock Fast algorithms for shortest paths in planar graphs, with
  applications.
\newblock {\em SIAM J. Comput.}, 16(6):1004--1022, 1987.

\bibitem{DiameterSODA18}
P.~Gawrychowski, H.~Kaplan, S.~Mozes, M.~Sharir, and O.~Weimann.
\newblock Voronoi diagrams on planar graphs, and computing the diameter in
  deterministic $\tilde {O}(n^{5/3})$ time.
\newblock {\em {SIAM} J. Comput.}, 50(2):509--554, 2021.

\bibitem{GawrychowskiMW21}
P.~Gawrychowski, S.~Mozes, and O.~Weimann.
\newblock Planar negative \emph{k}-cycle.
\newblock In {\em 32nd {SODA}}, pages 2717--2724, 2021.

\bibitem{ourSODA2018}
P.~Gawrychowski, S.~Mozes, O.~Weimann, and C.~Wulff-Nilsen.
\newblock Better tradeoffs for exact distance oracles in planar graphs.
\newblock In {\em SODA}, 2018.

\bibitem{GuX19}
Q.~Gu and G.~Xu.
\newblock Constant query time $(1+\varepsilon)$-approximate distance oracle for
  planar graphs.
\newblock {\em Theor. Comput. Sci.}, 761:78--88, 2019.

\bibitem{Geodesic}
J.~Hershberger and S.~Suri.
\newblock Matrix searching with the shortest-path metric.
\newblock {\em {SIAM} J. Comput.}, 26(6):1612--1634, 1997.

\bibitem{Husfeldt16}
T.~Husfeldt.
\newblock Computing graph distances parameterized by treewidth and diameter.
\newblock In {\em IPEC}, pages 16:1--16:11, 2016.

\bibitem{ItalianoKLS17}
G.~F. Italiano, A.~Karczmarz, J.~Lacki, and P.~Sankowski.
\newblock Decremental single-source reachability in planar digraphs.
\newblock In {\em 49th {STOC}}, pages 1108--1121, 2017.

\bibitem{insw-iamcm-11}
G.~F. Italiano, Y.~Nussbaum, P.~Sankowski, and C.~Wulff-Nilsen.
\newblock Improved algorithms for min cut and max flow in undirected planar
  graphs.
\newblock In {\em 43rd STOC}, pages 313--322, 2011.

\bibitem{Karczmarz18}
A.~Karczmarz.
\newblock Decrementai transitive closure and shortest paths for planar digraphs
  and beyond.
\newblock In {\em 29th {SODA}}, pages 73--92, 2018.

\bibitem{KawarabayashiKS11}
K.~Kawarabayashi, P.~N. Klein, and C.~Sommer.
\newblock Linear-space approximate distance oracles for planar, bounded-genus
  and minor-free graphs.
\newblock In {\em 38th {ICALP}}, volume 6755, pages 135--146, 2011.

\bibitem{KawarabayashiST13}
K.~Kawarabayashi, C.~Sommer, and M.~Thorup.
\newblock More compact oracles for approximate distances in undirected planar
  graphs.
\newblock In {\em 24th {SODA}}, pages 550--563, 2013.

\bibitem{Klein02}
P.~N. Klein.
\newblock Preprocessing an undirected planar network to enable fast approximate
  distance queries.
\newblock In {\em 13th {SODA}}, pages 820--827, 2002.

\bibitem{K05}
P.~N. Klein.
\newblock Multiple-source shortest paths in planar graphs.
\newblock In {\em SODA}, pages 146--155, 2005.

\bibitem{KMS13}
P.~N. Klein, S.~Mozes, and C.~Sommer.
\newblock Structured recursive separator decompositions for planar graphs in
  linear time.
\newblock In {\em 45th {STOC}}, pages 505--514, 2013.

\bibitem{LackiS17}
J.~Lacki and P.~Sankowski.
\newblock Optimal decremental connectivity in planar graphs.
\newblock {\em Theory Comput. Syst.}, 61(4):1037--1053, 2017.

\bibitem{LeW21}
H.~Le and C.~Wulff{-}Nilsen.
\newblock Optimal approximate distance oracle for planar graphs.
\newblock In {\em 62nd {IEEE} Annual Symposium on Foundations of Computer
  Science, {FOCS} 2021, Denver, CO, USA, February 7-10, 2022}, pages 363--374.
  {IEEE}, 2021.

\bibitem{LT79}
R.~J. Lipton and R.~E. Tarjan.
\newblock A separator theorem for planar graphs.
\newblock {\em SIAM J. Appl. Math}, 36(2):177--189, 1979.

\bibitem{LS11}
J.~\L\k{a}cki and P.~Sankowski.
\newblock Min-cuts and shortest cycles in planar graphs in ${O}(n \log \log n)$
  time.
\newblock In {\em 19th ESA}, pages 155--166, 2011.

\bibitem{LongPettie}
Y.~Long and S.~Pettie.
\newblock Planar distance oracles with better time-space tradeoffs.
\newblock In {\em Proceedings of the 32nd {ACM-SIAM} Symposium on Discrete
  Algorithms ({SODA})}, pages 2517--2536, 2021.

\bibitem{MozesS2012}
S.~Mozes and C.~Sommer.
\newblock Exact distance oracles for planar graphs.
\newblock In {\em 23rd {SODA}}, pages 209--222, 2012.

\bibitem{Nussbaum11}
Y.~Nussbaum.
\newblock Improved distance queries in planar graphs.
\newblock In {\em 12th {WADS}}, pages 642--653, 2011.

\bibitem{IntervalGraphs}
S.~Olariu.
\newblock A simple linear-time algorithm for computing the center of an
  interval graph.
\newblock {\em International Journal of Computer Mathematics}, 34:121--128,
  1990.

\bibitem{RV13}
L.~Roditty and V.~V. Williams.
\newblock Fast approximation algorithms for the diameter and radius of sparse
  graphs.
\newblock In {\em 45th {STOC}}, pages 515--524, 2013.

\bibitem{ubramanian93}
S.~Subramanian.
\newblock A fully dynamic data structure for reachability in planar digraphs.
\newblock In T.~Lengauer, editor, {\em 1st {ESA}}, volume 726, pages 372--383,
  1993.

\bibitem{Thorup04}
M.~Thorup.
\newblock Compact oracles for reachability and approximate distances in planar
  digraphs.
\newblock {\em J. {ACM}}, 51(6):993--1024, 2004.

\bibitem{WY16}
O.~Weimann and R.~Yuster.
\newblock Approximating the diameter of planar graphs in near linear time.
\newblock {\em {ACM} Trans. Algorithms}, 12(1):12:1 -- 12:13, 2016.

\bibitem{Wil04}
R.~Williams.
\newblock A new algorithm for optimal constraint satisfaction and its
  implications.
\newblock In {\em 31st {ICALP}}, pages 1227--1237, 2004.

\bibitem{WN08}
C.~Wulff-Nilsen.
\newblock Wiener index and diameter of a planar graph in subquadratic time.
\newblock Technical report, 08-16, Department of Computer Science, University
  of Copenhagen, 2008.
\newblock Available at
  \url{http://www.diku.dk/OLD/publikationer/tekniske.rapporter/rapporter/08-16.pdf}.
  Preliminary version in EurCG 2009.

\bibitem{Wulff-Nilsen10}
C.~Wulff-Nilsen.
\newblock {\em Algorithms for planar graphs and graphs in metric spaces}.
\newblock PhD thesis, University of Copenhagen, 2010.

\bibitem{Wulff-Nilsen16}
C.~Wulff-Nilsen.
\newblock Approximate distance oracles for planar graphs with improved query
  time-space tradeoff.
\newblock In {\em 27th {SODA}}, pages 351--362, 2016.

\end{thebibliography}

\end{document}